\documentclass[letterpaper,twocolumn,10pt]{article}
\usepackage{usenix2019_v3}

\usepackage{amsmath,amssymb,amsfonts,amsthm}
\usepackage{centernot}
\usepackage{graphicx}
\usepackage{booktabs}
\usepackage{tabularx}
\usepackage{array}
\usepackage{float}

\makeatletter
\newcommand{\FloatBarrier}{%
  \ifx\@deferlist\@empty\else
    \clearpage
  \fi}
\makeatother

\usepackage{listings}
\usepackage{algorithm}
\usepackage{algpseudocode}
\usepackage{subcaption}
\usepackage{xspace}
\usepackage{xcolor}
\usepackage{tikz}
\usetikzlibrary{arrows.meta, backgrounds, positioning, calc, shapes.geometric, fit, shadows, matrix, patterns, decorations.pathreplacing}
\microtypesetup{spacing=false}
\AtBeginDocument{\DeclareMathAlphabet{\mathcal}{OMS}{cmsy}{m}{n}}

\newcolumntype{L}[1]{>{\raggedright\arraybackslash}p{#1}}
\newcolumntype{C}[1]{>{\centering\arraybackslash}p{#1}}
\newcolumntype{R}[1]{>{\raggedleft\arraybackslash}p{#1}}

\newcommand{\cct}{\mbox{CCT}\xspace}
\newcommand{\sit}{\mbox{SIT}\xspace}

\newcommand{\lineagebench}{\mbox{IdentityLineageBench}\xspace}

\newtheorem{theorem}{Theorem}
\newtheorem{definition}{Definition}
\newtheorem{proposition}{Proposition}
\newtheorem{corollary}{Corollary}

\definecolor{slate}{RGB}{112,128,144}
\definecolor{emerald}{RGB}{16,163,127}
\definecolor{navy}{RGB}{24,49,83}
\definecolor{crimson}{RGB}{190,30,45}
\definecolor{amber}{RGB}{217,119,6}
\definecolor{cobalt}{RGB}{29,78,216}
\definecolor{teal}{RGB}{13,148,136}
\definecolor{purple}{RGB}{126,34,206}
\definecolor{softgray}{RGB}{247,249,252}
\definecolor{bordergray}{RGB}{218,224,233}
\definecolor{darkslate}{RGB}{60,72,88}
\definecolor{linegray}{RGB}{140,155,175}

\hypersetup{
  pdftitle={The Cognitive Continuity Test: Verifying Governed State Transitions in Persistent AI Agents},
  pdfauthor={Jun He; Deying Yu},
  pdfsubject={Verification and evaluation methodology for persistent cognitive identity transitions in autonomous agents},
  pdfkeywords={cognitive continuity test, persistent cognitive identity, transition witness, identity lineage, situated identity, IdentityLineageBench}
}

\begin{document}
\raggedbottom

\title{\bf The Cognitive Continuity Test:\\Verifying Governed State Transitions in Persistent AI Agents}

\ifdefined\CCTAnonymous
\author{Anonymous authors}
\hypersetup{pdfauthor={}}
\else
\author{
  {\rm Jun He}\\
  OpenKedge.io
  \and
  {\rm Deying Yu}\\
  OpenKedge.io
}

\fi

\maketitle

\begin{abstract}
Persistent AI agents revise beliefs, consolidate memory, and replace execution substrates. Similar successor states can accompany differently authorized transition claims, while legitimate development can change state substantially. We introduce the Cognitive Continuity Test (CCT), a policy-relative contract for verifying submitted transitions using scoped authority, provenance, deterministic application, semantic predicates, and candidate-persistence receipts. CCT distinguishes verified admissibility, affirmative violation, and unresolved required evidence. Separation results concern transition claims rather than live runtime identity; soundness is conditional on the specified checker and evaluator assumptions.

IdentityLineageBench provides 24 generated transition families. The reference post-resolution verifier matches all 576 canonical held-out labels; lexical state similarity and a lineage-only diagnostic baseline admit 60.0\% and 80.0\% of invalid fixtures. These comparisons establish synthetic conformance, not superiority to a policy-aware deployed system. Signed adversarial regressions cover fabricated interaction counts, unsupported belief changes, and mixed missing/contradictory evidence. SIT behavior and actual model migration remain unmeasured. An 18,000-execution valid-path study measures a 6.21\,ms default median on resident inputs. We specify the additional activation and recovery obligations needed for deployment.
\end{abstract}

\section{Introduction}
\label{sec:introduction}
Persistent AI agents update memory, revise beliefs, change software, and recover from failures. Their explicit state increasingly carries records of user instructions, prior actions, and access relationships~\cite{park2023generative,packer2023memgpt}. A system admitting a proposed successor must decide whether its changes are authorized and supported by the required evidence.

\subsection{The Verification Question}
\label{sec:intro-illusion}\label{sec:intro-question}
Behavioral similarity and transition admissibility answer different questions. An unauthorized principal can submit a copied state that retains the predecessor's memories and behavior. Conversely, authorized consolidation or evidence-supported belief revision can substantially change that state. Memory and personality benchmarks measure useful behavioral properties~\cite{wang2024incharacter,wu2025longmemeval}; they do not attempt to authenticate succession. The issue is therefore how to evaluate a \emph{submitted transition claim}, including its evidence, rather than how to replace those benchmarks.

Consider two claims proposing the same successor value. One carries authority for the declared operation; the other carries an authentic credential from an unauthorized principal. A state-only observation cannot distinguish these claims. This separation concerns the submitted records: copying the complete authorized state and witness preserves their historical audit validity. Detecting a live impersonator additionally requires runtime authentication and an activation protocol.

We introduce the Cognitive Continuity Test (\cct), a policy-relative contract for verifying explicit agent-state transitions. Given predecessor $X_t$, witness $\tau_t$, successor $X_{t+1}$, and trusted context, CCT returns \textsc{valid}, \textsc{invalid}, or \textsc{indeterminate}. A known contradiction yields \textsc{invalid}; unresolved required evidence yields \textsc{indeterminate} when no contradiction is established. This distinction supports retry and audit without admitting unsupported transitions.

\subsection{Scope and Contributions}
\label{sec:scope}\label{sec:intro-contributions}
CCT specializes established authorization, provenance, and replay mechanisms to typed identity-state changes. Its contributions are:
\begin{enumerate}
\item A non-circular witness contract binding proposal authority, dependencies, validation records, predecessor/successor digests, and a candidate-persistence receipt (Section~\ref{sec:system-model}).
\item Eight policy predicates and a ternary verification procedure, with separation results over submitted claims and conditional soundness under explicit checker and evaluator assumptions (Sections~\ref{sec:continuity-test}--\ref{sec:continuity-theory}). These results delimit the contract; they do not prove the implementation or an external evaluator sound.
\item \lineagebench, a generated suite of 24 transition families, identity-disjoint calibration, isolated checker ablations, and signed adversarial regression tests (Sections~\ref{sec:lineagebench}--\ref{sec:comparative-results}).
\item An executable \emph{post-resolution} verifier and reproducible classification artifacts. Historical policy retrieval/binding remains a caller obligation. The same reference implementation supplies the resident valid-path overhead study (Section~\ref{sec:verifier-overhead}).
\end{enumerate}
The supported claim is conformance of structured transition records to declared policy. Natural-language extraction quality, live-runtime exclusivity, operational crash recovery, and cross-model behavioral preservation require additional validation. Section~\ref{sec:related-work} identifies inherited mechanisms and the policy-aware comparisons still needed.

\section{Behavioral Identity and Transition Evidence}
\label{sec:triad}\label{sec:triad-definitions}
Human-referential fidelity asks whether an agent resembles a particular person. Situated identity (\sit) asks whether its behavior is consistent with a frozen developmental history~\cite{sit2026}. CCT instead evaluates the supplied transition's authority, provenance, and policy predicates. A state/response-mechanism pair can pass SIT while being accompanied by an unauthorized succession claim (Section~\ref{sec:sit-theory}). Merely copying stored history does not guarantee passing behavior.

\label{sec:sit-not-implies-cct}\label{sec:what-continuity-is-not}\label{sec:positive-continuity}
Continuity here is an operational, policy-relative relation over records. It permits learning and forgetting when the declared contract permits them. It neither establishes conscious or metaphysical identity nor guarantees that a runtime follows its stored norms. Authentic evidence can also record false external information; provenance integrity alone does not prevent memory poisoning.

\section{State, Witness, and Trust Model}
\label{sec:system-model}

\subsection{Explicit State and Mutations}
\label{sec:cognitive-state}
\begin{definition}[Cognitive Identity State]
\label{def:cognitive-state}
The explicit identity state at logical epoch $t$ is
\begin{equation}
X_t=\langle H_t,M_t,K_t,B_t,R_t,N_t,S_t\rangle.
\label{eq:state-tuple}
\end{equation}
\end{definition}
$H_t$ is the authenticated event chronicle; $M_t$ the derived, mutable memory store; $K_t$ acquired knowledge; $B_t$ beliefs with values, justifications, and revision history; $R_t$ relationships and disclosure metadata; $N_t$ governance and policy commitments; and $S_t$ explicit self-model fields. Model weights and the response-generating runtime are external to this tuple. Equal tuples do not denote distinct authenticated runtime instances.

Each event records content, provenance, and a timestamp in one policy-declared ordered domain, distinct from logical epochs. The historical-preservation preorder $H\le_{\mathrm{hist}}H'$ preserves prior event commitments, their order, and timestamps, either directly or through authorized tombstones. Working-memory compaction changes $M$ without erasing those commitments. Evidence establishes what was recorded, not whether an external assertion was objectively true.

\begin{equation}
d_t=\operatorname{Hash}(\mathsf{StateDomain},\operatorname{Enc}(X_t)).
\label{eq:state-digest}
\end{equation}
$\operatorname{Enc}$ is canonical encoding and every hash domain binds protocol version $v$. We abbreviate this state digest by $\operatorname{Hash}(X_t)$ (Appendix~\ref{app:envelope-metadata}).

\label{sec:state-mutations}
A mutation manifest $\Delta_t=(\delta_1,\ldots,\delta_m)$ contains typed operations with targets, arguments, and preconditions. Deterministic application computes the entire candidate or fails:
\begin{equation}
\widehat X_{t+1}=\operatorname{Apply}(X_t,\Delta_t).
\label{eq:delta-ops}
\end{equation}
A state-neutral \texttt{CLAIM\_SUCCESSION} is a nonempty operation requiring succession authority; unchanged payload does not make that authorization obligation vacuous. Runtime metadata such as execution epochs may be represented separately by an implementation adapter.

\subsection{The Cognitive Transition Witness}
\label{sec:transition-witness}
The proposal and its digest are
\begin{equation}
\begin{aligned}
q_t&=\langle\mathit{id},t,p_t,\Delta_t,A_t,E_t\rangle,\\
d_t^{\mathrm{prop}}&=\operatorname{Hash}(\mathsf{ProposalDomain},\operatorname{Enc}(q_t)),
\end{aligned}
\label{eq:proposal}
\end{equation}
where $p_t=d_t$. Authority bundle $A_t$ contains principal, role, scope, and signature records. Each signature binds identity, epoch, predecessor, ordered mutations, provenance, and its signer's principal/role/scope, excluding $A_t$ itself. $E_t$ names input dependencies, source citations, and temporal evidence; the trusted resolver supplies their contents independently of the candidate.

Validators bind their result to the proposal:
\begin{equation}
\begin{aligned}
V_t=\{\langle &\mathit{check},\mathit{evaluator},\mathit{ver},\mathit{policy\_ver},\\
&d_t^{\mathrm{prop}},\mathit{result},s,\sigma^{\mathrm{eval}}\rangle\}.
\end{aligned}
\label{eq:validation-records}
\end{equation}
Native checks are recomputed. External checks require authorized evaluators, compatible versions, fresh scoped signatures, and a policy-defined quorum. Authentic signatures establish attribution; correctness of external semantic judgments remains a separate assumption.

\begin{equation}
\begin{aligned}
\tau_t^{\mathrm{core}}&=\langle q_t,V_t\rangle,\\
d_t^{\mathrm{core}}&=\operatorname{Hash}(\mathsf{WitnessCoreDomain},\operatorname{Enc}(\tau_t^{\mathrm{core}})).
\end{aligned}
\label{eq:witness-core}
\end{equation}
A trusted persistence service then signs
\begin{equation}
\begin{aligned}
R_t^{\mathrm{commit}}=\langle &\mathit{rcpt\_id},\mathit{kernel\_id},\mathit{id},t,p_t,\\
&d_{t+1},d_t^{\mathrm{core}},s_{\mathrm{commit}},\sigma_{\mathrm{kernel}}\rangle.
\end{aligned}
\label{eq:commit-receipt}
\end{equation}
The signature covers every preceding field. The completed witness is
\begin{equation}
\tau_t=\langle\tau_t^{\mathrm{core}},R_t^{\mathrm{commit}}\rangle.
\label{eq:transition-witness}
\end{equation}
No object contains its own digest. The retained wire name \texttt{commit\_receipt} denotes \emph{candidate persistence}, including the finalized validation records. It neither activates a branch head nor certifies semantic admissibility. This receipt must not be conflated with the Continuity Kernel's authoritative activation transaction~\cite{ck2026}. CCT may reject a correctly persisted candidate. Section~\ref{sec:disc-activation} states the separate activation obligations.

\subsection{Historical Policy and Trusted Inputs}
\label{sec:continuity-policy}
A policy $\Pi_t$ specifies lineage, authority, history, time, belief, relationship, normative, and provenance predicates. The full contract requires pre-state governance to bind an immutable reference
\begin{equation}
\begin{aligned}
N_t.\mathit{policy\_ref}&=\langle\mathit{id}_{\Pi},\mathit{ver}_{\Pi},d_{\Pi}\rangle,\\
d_{\Pi}&=\operatorname{Hash}(\mathsf{PolicyDomain},\operatorname{Enc}(\Pi_t)).
\end{aligned}
\label{eq:policy-digest}
\end{equation}
The reference is covered by $p_t$. A trusted resolver $\Omega$ retrieves the artifact, checks its canonical digest and version compatibility, and returns
\begin{equation}
\operatorname{ResolvePolicy}(N_t.\mathit{policy\_ref},v,\Omega)
=\langle\widehat\Pi_t,s_{\mathrm{pol}}\rangle.
\label{eq:resolve-policy}
\end{equation}
A match yields $\widehat\Pi_t=\Pi_t$ and \textsf{PASS}; unavailability yields $\bot,\mathsf{UNKNOWN}$; a present digest/version contradiction yields $\bot,\mathsf{FAIL}$. Later amendments do not reinterpret old transitions. \emph{The executable artifact accepts a caller-supplied resolved policy; this historical binding interface is specified but not implemented or measured.}

Trusted context $\Gamma=(\mathit{id}^*,t^*,v,\mathit{branch\_evidence})$ fixes the claim's scope. Credential histories, acquisition records, and any required head evidence must be independently authenticated. An attacker may submit malformed or unauthorized proposals and validly signed but inadmissible mutations. The threat model does not assume that signatures make their contents true. A malicious resolver, dishonest persistence service, or unsound semantic evaluator violates explicit trust assumptions. A historical key-status check also requires authenticated temporal anchoring; a signer's self-declared timestamp cannot exclude backdating after compromise.

\subsection{Declarative Admissibility}
\label{sec:admissible-relation}
\begin{definition}[Admissible Transition Relation]
\label{def:admissible-transition}
For fixed $\Gamma,\Omega$ and successfully resolved $\Pi_t$, write
\begin{equation}
X_t\overset{\tau_t}{\rightsquigarrow}_{\Pi_t}X_{t+1}
\label{eq:admissible-relation}
\end{equation}
exactly when (1) predecessor, identity, epoch, protocol version, and policy-required head evidence agree; (2) $X_{t+1}=\operatorname{Apply}(X_t,\Delta_t)$; (3) pre-state governance authorizes every operation; (4) policy-required dependencies resolve with matching content commitments; (5) all declared semantic predicates $\mathcal P_k$, $k\in\{\mathrm{hist,temp,bel,rel,norm}\}$ hold; and (6) a historically valid persistence receipt binds the scope and all three relevant digests.
\end{definition}
This relation labels the \emph{supplied witness}. The shorthand $X_t\rightsquigarrow_{\Pi_t}X'\iff\exists\tau:X_t\overset{\tau}{\rightsquigarrow}_{\Pi_t}X'$ asks a different, existential question about state values. Neither establishes possession by a live runtime. Operational checks below implement the clauses relative to their declared evidence and trust interfaces.

\section{The Cognitive Continuity Test}
\label{sec:continuity-test}
\subsection{Ternary Decision Contract}
\label{sec:cct-definition}
\begin{definition}[Checker Result Domain]
\label{def:checker-domain}
A checker returns $\mathsf{PASS}$ for a verified predicate, $\mathsf{FAIL}$ for affirmative violation, and $\mathsf{UNKNOWN}$ for insufficient required evidence. Missing evidence and a present contradiction are distinct conditions.
\end{definition}

\begin{definition}[CCT Aggregation]
\label{def:cct}
For fixed mandatory states $X_t,X_{t+1}$ and trusted contexts, let $C_{\mathrm{pol}}$ be the result of historical policy resolution. Define the evaluated check set
\begin{equation}
\begin{aligned}
\mathcal C=\{&C_{\mathrm{pol}},C_{\mathrm{apply}},I_{\mathrm{lin}},I_{\mathrm{auth}},I_{\mathrm{prov}},\\
&I_{\mathrm{hist}},I_{\mathrm{temp}},I_{\mathrm{bel}},I_{\mathrm{rel}},I_{\mathrm{norm}}\}.
\end{aligned}
\label{eq:verifier-check-set}
\end{equation}
$C_{\mathrm{apply}}$ verifies $\operatorname{Apply}(X_t,\Delta_t)=X_{t+1}$. Missing $\Delta_t$ yields $\mathsf{UNKNOWN}$; malformed operations, unsatisfied preconditions, undefined application, or a successor mismatch yield $\mathsf{FAIL}$. The implementation compares canonical state digests, relying on cryptographic binding.
Let $\mathcal F$ and $\mathcal U$ be the failed and unknown checks.
\begin{equation}
\operatorname{CCT}_{\Pi_t}=\begin{cases}
\mathsf{INVALID}&\mathcal F\ne\varnothing,\\
\mathsf{INDETERMINATE}&\mathcal F=\varnothing,\ \mathcal U\ne\varnothing,\\
\mathsf{VALID}&\mathcal F=\mathcal U=\varnothing.
\end{cases}
\label{eq:cct-aggregation}
\end{equation}
\end{definition}
A present independent contradiction dominates missing evidence, including within composite checkers. Only \textsc{valid} can proceed to the separate activation protocol (Section~\ref{sec:disc-activation}); a binary admission gate can enforce the same acceptance boundary, while the third verdict distinguishes retryable uncertainty from proved invalidity.

\subsection{Eight Invariant Classes}
\label{sec:invariants}
The predicates below define the full contract. Every checker operates on the fixed states, witness, policy projection, and applicable trusted evidence.

\paragraph{Lineage ($I_{\mathrm{lin}}$).}
Proposal identity/epoch/version must match $\Gamma$; $p_t=\operatorname{Hash}(X_t)$. The receipt must bind that scope and the predecessor, successor, and finalized-core digests with an authorized historically valid signature. Ordinary later key rotation does not invalidate earlier signatures; unknown historical status yields $\mathsf{UNKNOWN}$. Required head evidence must match the scoped proposal and successor. A contradictory digest yields $\mathsf{FAIL}$ even when exclusivity is unresolved. Head evidence concerns the transition epoch, not current runtime freshness.
\label{eq:inv-lineage}

\paragraph{Authority ($I_{\mathrm{auth}}$).}
Every operation requires scoped credentials under pre-state $N_t$:
\begin{equation}
\forall\delta\in\Delta_t:\quad
\operatorname{Authorized}(\delta,q_t,N_t,\Pi_{t,\mathrm{auth}}).
\label{eq:inv-auth}
\end{equation}
Missing credentials yield $\mathsf{UNKNOWN}$; an excluded signer or insufficient established scope yields $\mathsf{FAIL}$. Authority to invoke an amendment does not override a non-amendable invariant. Succession claims use a nonempty operation even when the cognitive payload is unchanged.

\paragraph{History ($I_{\mathrm{hist}}$).}
\begin{equation}
H_t\le_{\mathrm{hist}}H_{t+1}\ \land\
\neg\operatorname{ForbiddenHist}(H_t,\Delta_t,H_{t+1},\Pi_{t,\mathrm{hist}}).
\label{eq:inv-hist}
\end{equation}
Prior commitments, order, and timestamps must survive directly or as authorized tombstones. New events require authenticated acquisition support. A contradictory source or complete authenticated index establishes fabrication; an unavailable acquisition record only establishes uncertainty. Corrections append evidence rather than silently rewriting prior events.

\paragraph{Time ($I_{\mathrm{temp}}$).}
Let $a_{t+1}(\chi)$ be a knowledge item's claimed acquisition time, $s_{\mathrm{acq}}$ its authenticated acquisition time, $\operatorname{New}(H_t,H_{t+1})$ the newly inserted events, and $h(H)=\max_{e\in H}s(e)$ including tombstone timestamps. Require
\begin{equation}
\begin{gathered}
\forall\chi\in K_{t+1}:\quad a_{t+1}(\chi)\ge s_{\mathrm{acq}}(\chi),\\
\forall e\in\operatorname{New}(H_t,H_{t+1}):\quad s(e)\ge s_{\mathrm{acq}}(e),\\
h(H_t)\le h(H_{t+1}).
\end{gathered}
\label{eq:inv-temp}
\end{equation}
Set $h(\varnothing)=\bot$. Missing required acquisition evidence yields $\mathsf{UNKNOWN}$; affirmative backdating yields $\mathsf{FAIL}$. The per-event condition is necessary: a maximum alone cannot detect inserting an older event. The native knowledge checker covers structured entries carrying \texttt{claimed\_at}; untyped knowledge entries are outside its temporal guarantee.

\paragraph{Beliefs ($I_{\mathrm{bel}}$).}
For every new or changed complete record $B[\phi]$, including metadata-only changes, require
\begin{equation}
\operatorname{JustifiedRevision}(\phi)\land
\operatorname{HistoricalBeliefPreserved}(\phi).
\label{eq:inv-bel}
\end{equation}
These predicates depend on both states, $\Delta_t,E_t,V_t$, and $\Pi_{t,\mathrm{bel}}$. Prior origin and revision history must remain represented. The native policy requires scoped evidence matching proposition, value, and confidence, with revision epochs consistent with the applied transition. A citation's existence is insufficient; absent support yields $\mathsf{UNKNOWN}$ and contradictory support yields $\mathsf{FAIL}$. External semantic attestations may replace native checking only under their declared trust assumption.

\paragraph{Relationships ($I_{\mathrm{rel}}$).}
\begin{equation}
\begin{gathered}
\forall u\text{ with }\Delta R(u)\ne\varnothing:\\
\operatorname{RelationalSupport}(u,\Delta R(u),E_t,A_t,\Pi_{t,\mathrm{rel}}).
\end{gathered}
\label{eq:inv-rel}
\end{equation}
The native policy checks changed relationships, including count-only changes, against scoped positive interaction records. Session identifiers must be unique and append-only, counts must equal their cardinality, and elevation requires sufficient established history. Fabricating a counter cannot create evidence for a later elevation. Missing support yields $\mathsf{UNKNOWN}$; contradictory records or counts yield $\mathsf{FAIL}$. Bilateral-authorization alternatives require another policy implementation.

\paragraph{Norms ($I_{\mathrm{norm}}$).}
$\operatorname{NormativeAdmissible}$ compares the old and new governance/self-model records under $\Pi_{t,\mathrm{norm}}$, enforcing amendment rules and designated non-amendable anchors. The native artifact protects explicit non-amendable normative records; general self-model anchor policies remain part of the broader specification.
\label{eq:inv-norm}

\paragraph{Provenance ($I_{\mathrm{prov}}$).}
$\operatorname{DependencyResolved}(\delta,E_t,\Pi_{t,\mathrm{prov}})$ must hold for every policy-required dependency of each operation. Obligations derive from operation type and policy, not merely candidate-supplied reference fields. Missing required dependencies yield $\mathsf{UNKNOWN}$; resolved content contradicting a commitment yields $\mathsf{FAIL}$.
\label{eq:inv-prov}

\subsection{Verification Algorithm and Implemented Scope}
\label{sec:cct-algorithm}
Algorithm~\ref{alg:cct-verification} specifies the full contract. Partial parsing tags fields as present, absent, or malformed; a complete canonical preimage is required before hashing. Checkers evaluate independent subpredicates and retain failures even when another prerequisite is absent. Policy projections of $\bot$ remain unavailable. Unexpected runtime exceptions are implementation defects, not policy verdicts.

The Python artifact implements a structured post-resolution subset of Stages~1--5 with a supplied trusted policy and resident resolver snapshots. It accepts schema-valid typed objects and supported optional evidence omissions; general \textsc{ParsePartial} handling is specified but unimplemented. An omitted mutation manifest defaults to an empty list in the schema, so the specification's absent-manifest distinction is not exercised by this interface. It does not implement Stage~0, network authentication of resolver contents, distributed receipt issuance, or activation. The theorem's native/evaluator soundness assumptions apply to whichever implementation is deployed.

\begin{algorithm}[t]
\caption{\textsc{VerifyCognitiveContinuity}}
\label{alg:cct-verification}
\footnotesize
\begin{algorithmic}[1]
\State \textbf{Input:} $X_t, \tau_t, X_{t+1}, \Gamma, \Omega$. \textbf{Output:} $\operatorname{Verdict} \in \{\mathsf{VALID}, \mathsf{INVALID}, \mathsf{INDETERMINATE}\}$, Report $\mathcal{R}$.
\State $\mathcal{C}\gets\{\}$ \Comment{map from required check to ternary result and reason}
\Statex \textbf{// Stage 0: Deterministic Historical Policy Resolution}
\State $\langle\widehat{\Pi}_t,\mathcal{C}[C_{\mathrm{pol}}]\rangle\gets\operatorname{ResolvePolicy}(X_t.N.\mathit{policy\_ref},v,\Omega)$
\State $Z\gets\operatorname{ParsePartial}(\tau_t)$ \Comment{tagged $q_t,\Delta_t,A_t,E_t,V_t,R_t,v_\tau,b_t,p_t$}
\Statex \textbf{// Stage 1: Lineage and Cryptographic Integrity ($I_{\mathrm{lin}}$)}
\State $\mathcal{C}[I_{\mathrm{lin}}]\gets\operatorname{CheckLineagePartial}(X_t,Z,X_{t+1},\operatorname{Proj}_{\mathrm{lin}}(\widehat{\Pi}_t),\Gamma)$
\Statex \textbf{// Stage 2: Authority and Governance Verification ($I_{\mathrm{auth}}$)}
\State $\mathcal{C}[I_{\mathrm{auth}}]\gets\operatorname{CheckAuthorityPartial}(Z.q_t,X_t.N,\operatorname{Proj}_{\mathrm{auth}}(\widehat{\Pi}_t))$
\Statex \textbf{// Stage 3: Provenance Completeness ($I_{\mathrm{prov}}$)}
\State $\mathcal{C}[I_{\mathrm{prov}}]\gets\operatorname{CheckProvenancePartial}($
\Statex \hspace{\algorithmicindent}$Z.E_t,Z.\Delta_t,\operatorname{Proj}_{\mathrm{prov}}(\widehat{\Pi}_t))$
\Statex \textbf{// Stage 4: Deterministic State Application ($C_{\mathrm{apply}}$)}
\State $\mathcal{C}[C_{\mathrm{apply}}]\gets\operatorname{CheckApplyPartial}(X_t,Z.\Delta_t,X_{t+1})$
\Statex \textbf{// Stage 5: Semantic Invariant Evaluation and Attestations}
\For{each invariant $I_k \in \{I_{\mathrm{hist}}, I_{\mathrm{temp}}, I_{\mathrm{bel}}, I_{\mathrm{rel}}, I_{\mathrm{norm}}\}$}
    \State $\widehat{\Pi}_{t,k}\gets\operatorname{Proj}_k(\widehat{\Pi}_t)$
    \If{$\widehat{\Pi}_{t,k}\ne\bot\ \land\ \widehat{\Pi}_{t,k}.\text{is\_attested}$}
        \State $\mathcal{C}[I_k]\gets\operatorname{VerifyAttestationPartial}($
        \Statex $Z.V_t,I_k,Z.d_t^{\mathrm{prop}},\widehat{\Pi}_{t,k})$
    \Else
        \State $\mathcal{C}[I_k]\gets\operatorname{EvaluateInvariantPartial}($
        \Statex $X_t,Z.\tau_t^{\mathrm{core}},X_{t+1},\widehat{\Pi}_{t,k})$
    \EndIf
\EndFor
\State \Return $\operatorname{Aggregate}(\mathcal{C})$ \Comment{FAIL $>$ UNKNOWN $>$ PASS; retain all reasons}
\end{algorithmic}
\end{algorithm}

\section{Continuity Theory}
\label{sec:continuity-theory}

The following results delimit what observations can establish about a submitted transition and state the assumptions needed for verifier soundness. They concern historical transition verification; authenticating a live runtime and enforcing exclusive execution require additional mechanisms.

\subsection{Successor-State Equivalence Insufficiency}
\label{sec:state-insufficiency}

Fix a predecessor $X_t$, context $\Gamma$, and resolved policy $\Pi_t$. A \emph{submitted transition claim} is $c=\langle X_t,\tau,X'\rangle$. Write $\operatorname{Adm}(c)$ for the witness-indexed relation $X_t\overset{\tau}{\rightsquigarrow}_{\Pi_t}X'$. This labels the supplied claim, not the existential question of whether \emph{some} witness could authorize the state value $X'$. Let $\mathcal{O}_s:\mathcal{X}\to\mathcal{Y}$ observe a successor state, and define the claim observation $\mathcal{O}(c)=\mathcal{O}_s(X')$. Even full state inspection omits the claim's witness.

\begin{definition}[State-Only Continuity Test]
A state-only evaluator is a deterministic or randomized decision rule $T:\mathcal{Y}\to\Delta(\mathcal{D})$, where $\mathcal{D}=\{\mathsf{VALID},\mathsf{INVALID},\mathsf{INDETERMINATE}\}$, whose only candidate-dependent input is $\mathcal{O}(c)$.
\end{definition}

\begin{theorem}[Successor-State Equivalence Insufficiency]
\label{thm:state-insufficiency}
Suppose two submitted claims $c_+$ and $c_-$ satisfy
\begin{equation}
\mathcal{O}(c_+)=\mathcal{O}(c_-),\qquad
\operatorname{Adm}(c_+)\land\neg\operatorname{Adm}(c_-).
\label{eq:state-equiv-premise}
\end{equation}
Then every state-only evaluator has the same decision distribution on both claims. In particular, it cannot both assign zero probability of $\mathsf{VALID}$ to $c_-$ and positive probability of $\mathsf{VALID}$ to $c_+$.
\end{theorem}

\begin{proof}
Both inputs equal the same observation $y$. Thus, for each $d\in\mathcal{D}$,
\begin{equation}
\Pr[T(\mathcal{O}(c_+))=d]=\Pr[T(\mathcal{O}(c_-))=d].
\end{equation}
Taking $d=\mathsf{VALID}$ gives the claim. Abstaining on both avoids false validation but does not validate $c_+$.
\end{proof}

The premise is realizable even for full-state observation. Choose a policy permitting the nonempty, state-neutral operation $\Delta=(\texttt{CLAIM\_SUCCESSION})$, with $\operatorname{Apply}(X_t,\Delta)=X_t$. Submit two otherwise well-formed claims with this same successor value: one has scoped credentials from an authorized principal, and the other has credentials from a principal affirmatively excluded by governance. Each has its own authentic receipt binding its finalized core; all non-authority predicates hold. The first claim is admissible and the second is not. Receipt issuance here certifies persistence of the submitted candidate, so it need not certify authorization. Appendix~\ref{app:proof-thm1} details the construction.

This example does not make the same state value both existentially admissible and inadmissible. Nor does it distinguish a live clone that copies the state \emph{and the complete authorized historical witness}: an offline verifier receives identical records in that case.

\begin{corollary}[Transition Evidence Requirement]
\label{cor:transition-evidence}
Under Eq.~\eqref{eq:state-equiv-premise}, an evaluator that separates the two claims must receive additional input that distinguishes them. A transition witness can supply such evidence; the theorem does not require CCT's particular witness format or establish its sufficiency.
\end{corollary}

\subsection{Situated Identity Does Not Imply Continuity}
\label{sec:sit-theory}

\begin{proposition}[$\sit \centernot\implies \cct$]
\label{prop:sit-not-cct}
Consider a fixed SIT protocol $\mathcal{B}$ that admits a passing state/response-mechanism pair $(X,f)$ and does not inspect succession credentials. Under a policy permitting state-neutral \texttt{CLAIM\_SUCCESSION} by authorized principals, a claim with successor $X$ and response mechanism $f$ can satisfy
\begin{equation}
\begin{gathered}
\operatorname{SIT}_{\mathcal{B}}(X,f)=\mathsf{PASS},\\
\operatorname{CCT}_{\Pi_t}(X,\tau_-,X)=\mathsf{INVALID}.
\end{gathered}
\label{eq:sit-not-cct-prop}
\end{equation}
\end{proposition}

\begin{proof}
Keep both $X$ and the passing response mechanism $f$ fixed. Submit $\Delta=(\texttt{CLAIM\_SUCCESSION})$ with affirmatively unauthorized scoped credentials, matching digests, an authentic candidate-persistence receipt, and all required non-authority evidence. The behavioral inputs and mechanism remain unchanged, so SIT still passes. The nonempty operation requires authority, yielding $I_{\mathrm{auth}}=\mathsf{FAIL}$ and therefore $\mathsf{INVALID}$ by Eq.~\eqref{eq:cct-aggregation}.
\end{proof}

The passing response mechanism is an explicit premise: copying a chronicle alone does not guarantee correct retrieval, ignorance, or disclosure behavior. This separation is not an empirical SIT result or a mechanism for authenticating the runtime presenting $X$.

\subsection{Cryptographic Lineage Insufficiency}
\label{sec:lineage-insufficiency}

\begin{proposition}[Authentic Ancestry $\neq$ Admissible Continuity]
\label{prop:crypto-insufficiency}
Authentic lineage and sufficient mutation-class authority need not imply semantic admissibility:
\begin{equation}
\begin{gathered}
(I_{\mathrm{lin}}=\mathsf{PASS})\land(I_{\mathrm{auth}}=\mathsf{PASS})\\
\centernot\implies X_t\overset{\tau_t}{\rightsquigarrow}_{\Pi_t}X_{t+1}.
\end{gathered}
\end{equation}
\end{proposition}

\begin{proof}
Let $\psi\in N_{\mathrm{core}}$ be non-amendable under $\Pi_{t,\mathrm{norm}}$. Submit $\Delta_t=(\operatorname{DeleteNormativeConstraint}(\psi))$ with authority to invoke normative amendments and an authentic receipt binding the correctly applied successor and finalized witness core. Choose complete evidence so all non-normative predicates hold. Lineage and class-level authority pass, while deleting $\psi$ yields $I_{\mathrm{norm}}=\mathsf{FAIL}$. Persistence and authority to submit amendments do not imply permission for every amendment.
\end{proof}

\subsection{Conditional CCT Verifier Soundness}
\label{sec:sufficiency-theory}

\begin{theorem}[Conditional CCT Verifier Soundness]
\label{thm:sufficiency}
Fix $X_t,\tau_t,X_{t+1},\Gamma,\Omega$ and a successfully resolved historical policy $\Pi_t$. Assume:
\begin{enumerate}\setlength{\itemsep}{1.5pt}\setlength{\parskip}{0pt}\setlength{\parsep}{0pt}
    \item \textbf{A1 (Cryptographic Binding and Unforgeability):} Canonical, versioned, domain-separated digests resist collisions and second preimages; scoped signatures resist chosen-message forgery.
    \item \textbf{A2 (Deterministic State Application):} The executor implements the specified deterministic mutation semantics.
    \item \textbf{A3 (Native Verification and Resolution Soundness):} Successful policy resolution establishes commitment and version compatibility, and each native $\mathsf{PASS}$ implies its declarative predicate. This includes the storage and head-evidence predicates: an authentic receipt signature alone does not prove storage behavior.
    \item \textbf{A4 (Attestation Resolution Soundness):} A policy-authorized attestation/quorum result of $\mathsf{PASS}$ for $I_k$ implies the semantic predicate $\mathcal{P}_k$, in addition to authenticating the signed record.
\end{enumerate}
Then $\operatorname{CCT}_{\Pi_t}(X_t,\tau_t,X_{t+1})=\mathsf{VALID}$ implies $X_t\overset{\tau_t}{\rightsquigarrow}_{\Pi_t}X_{t+1}$, except with negligible cryptographic probability.
\end{theorem}

\begin{proof}
A $\mathsf{VALID}$ verdict requires every result in $\mathcal{C}$ to be $\mathsf{PASS}$. A1--A3 map policy resolution, lineage, authority, provenance, and application checks to their declarative clauses; A3 or A4 supplies each semantic predicate. Their conjunction is Definition~\ref{def:admissible-transition}. Appendix~\ref{app:proof-thm2} gives the clause-by-clause argument.
\end{proof}

This is a conditional composition guarantee. A3 and A4 contain the substantive checker and evaluator correctness obligations; the theorem does not prove them for the reference implementation, and fixture agreement cannot discharge them. With imperfect evaluators, a negligible cryptographic error bound is insufficient: their semantic error also contributes. No completeness, availability, live-instance exclusivity, or crash-recovery property follows from this theorem.

\subsection{Compositional Lineage Continuity}
\label{sec:compositional-theory}

\begin{definition}[Multi-Step Lineage Continuation]
\label{def:multistep-continuation}
For $n\ge1$, let $\vec{\tau}=(\tau_0,\dots,\tau_{n-1})$ describe one identity lineage. Each step has its expected epoch and context $\Gamma_k$, and uses the historical policy $\Pi_k$ successfully resolved from $N_k$. Write
\begin{equation}
X_0\overset{\tau_0}{\rightsquigarrow}_{\Pi_0}X_1\cdots
\overset{\tau_{n-1}}{\rightsquigarrow}_{\Pi_{n-1}}X_n
\label{eq:transition-chain}
\end{equation}
and define the nonempty witness-indexed closure
\begin{equation}
X_0\overset{\vec{\tau}}{\rightsquigarrow}^{*}X_n
\iff \forall k<n,\quad X_k\overset{\tau_k}{\rightsquigarrow}_{\Pi_k}X_{k+1}.
\label{eq:transitive-closure}
\end{equation}
\end{definition}

\begin{proposition}[Preservation of Transitive Invariants]
\label{prop:compositional}
Let $n\ge1$ and $X_0\overset{\vec{\tau}}{\rightsquigarrow}^{*}X_n$. If every admissible step implies $\mathcal{P}(X_k,X_{k+1})$ and $\mathcal{P}$ is a transitive binary relation, then $\mathcal{P}(X_0,X_n)$.
\end{proposition}

Induction on $n$ proves the proposition (Appendix~\ref{app:proof-prop3}). Consequences include a receipt ancestry path from $d_0$ to $d_n$, event-time maxima $h(H_0)\le\cdots\le h(H_n)$, and historical preservation $H_0\le_{\mathrm{hist}}H_n$. For fixed anchors $\mathcal{A}_{\mathrm{core}}$, if every resolved policy protects them and each step preserves them, their inclusion in $N_{\mathrm{core}}\cup S_{\mathrm{core}}$ persists throughout the lineage. This assumes protection under every policy; it does not prove resistance to gradual policy weakening. Non-transitive semantic bounds, branch reconciliation, and long-horizon evaluator error require separate analysis.

\section{IdentityLineageBench}
\label{sec:lineagebench}
\label{sec:bench-arch}

\lineagebench generates synthetic transition records
\begin{equation}
\mathcal{T}_k=\langle X_t,\tau_t,X_{t+1},y_k^*\rangle,
\label{eq:benchmark-instance}
\end{equation}
with labels derived from the generator, frozen policy $\Pi$, and verification context $\Gamma$. Labels express specification conformance, not external truth or observed agent behavior.

\paragraph{Taxonomy.}
\label{sec:bench-taxonomy}
The 24 families cover ten legitimate changes (L1--L10), ten invalid mutations (I1--I10), and four evidence gaps (D1--D4). Legitimate cases include learning, episodic formation, belief revision, trust growth, consolidation, forgetting, correction, runtime upgrades, recovery records, and substrate rebinding. Invalid cases exercise fabrication, retrospective rewriting, forbidden amendments, trust injection, suppression, replay mismatch, unauthorized succession, rollback, receipt forgery, and corrupted provenance. Evidence gaps concern required attestations, head evidence, citations, and evaluator agreement. Appendix~\ref{app:transition-taxonomy} specifies each family.

L9 represents a forward recovery record after authenticated replay reconstructs the accepted suffix; checkpoint recovery and suffix replay are not executed. L10 changes a synthetic model identifier under a migration manifest; no foundation model is run or migrated.

\paragraph{Generation.}
\label{sec:bench-generation}
Source lifecycles advance through authorized episodic appends (L2); other families branch from these source states. The benchmark therefore exercises multiple epochs but not long trajectories combining all mutation families. The native belief fixtures require support scoped to identity, epoch, proposition, value, and confidence. Relational fixtures use unique authenticated session records and derive interaction counts from them, including changes to counts without trust elevation. Missing required evidence yields $\mathsf{UNKNOWN}$; an established contradiction yields $\mathsf{FAIL}$ and dominates concurrent evidence gaps.

\section{Evaluation Methodology}
\label{sec:experimental-methodology}
\label{sec:baseline-arms}

The evaluation defines six arms:
\begin{description}\setlength{\itemsep}{0pt}
\item[A: State similarity.] Thresholded cosine similarity of a deterministic lexical token-count projection of canonical state JSON. This is not a semantic encoder.
\item[B: Memory overlap.] Thresholded $|M_t\cap M_{t+1}|/|M_t|$ for nonempty predecessor memory, using exact canonical items.
\item[C: SIT.] A proposed situated-response comparison~\cite{sit2026}. No frozen model, probe, candidate, or judge configuration was supplied; this arm is unexecuted.
\item[D: Lineage only.] Ternary verification of predecessor, scope, receipt, historical credentials, and required head evidence. It omits mutation authority and semantic checks.
\item[E: CCT.] The post-resolution checks of Algorithm~\ref{alg:cct-verification}. The benchmark supplies the frozen resolved policy directly, fixing $C_{\mathrm{pol}}=\mathsf{PASS}$; resolver failures are unmeasured.
\item[F: Fixture oracle.] Generator labels, used to check scoring rather than as an independent comparator.
\end{description}

Arms A, B, and D are restricted diagnostic proxies. Their errors identify information omitted by those configurations; they do not establish superiority over a capable policy-aware system. A and B emit binary predictions and cannot abstain; D can abstain only within its lineage scope.

\paragraph{Calibration and questions.}
\label{sec:rqs}
Development and test identities, states, histories, memories, and derived fixtures are disjoint. Thresholds maximize balanced accuracy on development canonical valid/invalid fixtures; indeterminate fixtures are excluded from calibration. Held-out identities share the same generator and mutation families. The measured questions concern fixture conformance, evidence gaps, isolated checker contribution, and local verification cost; behavioral continuity remains unmeasured.

\paragraph{Metrics.}
\label{sec:metrics}
Write $V$, $I$, and $U$ for valid, invalid, and indeterminate verdicts. We report
\begin{align*}
\mathrm{VAR}&=P(\widehat y=V\mid y^*=V),\\
\mathrm{IAR}&=P(\widehat y=V\mid y^*=I),\\
\mathrm{IDR}&=P(\widehat y=I\mid y^*=I),\\
\mathrm{IDR}_{\mathrm{indet}}&=P(\widehat y=U\mid y^*=U).
\end{align*}
Valid nonacceptance $\mathrm{VNAR}=1-\mathrm{VAR}$ is partitioned into rejection (VRR) and abstention (VIR). Family and checker rates accompany aggregate scores. Section~\ref{sec:verifier-overhead} measures local verifier latency and serialized witness size with parsed resident inputs, excluding storage-service, network, evaluator, and model costs.

\section{Fault Isolation and Execution}
\label{sec:evaluation-protocol}
\label{sec:adversarial-attacks}
\label{sec:empirical-hypotheses}

Canonical I1 changes autobiographical content while retaining a nonbackdated timestamp and contradictory authenticated acquisition record, isolating history failure. I7 supplies affirmatively unauthorized succession credentials with valid lineage; copying a complete authorized record is not the tested attack. I8 intentionally combines stale lineage and temporal regression. Canonical families, isolated faults, and compound variants are reported separately.

\paragraph{Execution.}
\label{sec:protocol-execution}
Each available arm evaluates the same frozen records independently. Checker ablations remove one condition on fixtures where only that condition is decisive; compound failures cannot attribute marginal contribution. Classification uses generator labels and reports fixture-bootstrap 95\% intervals, which characterize the sampled fixtures rather than generalization across deployments. Source/configuration digests and prediction records support reproduction. Section~\ref{sec:comparative-results} reports the measurements.

\section{Reference Verifier Overhead Characterization}
\label{sec:verifier-overhead}
\subsection{Workloads and Measurement Protocol}
\label{sec:overhead-method}
We measure the same reference implementation used in Section~\ref{sec:comparative-results}. The study covers 60 deterministic signed workloads, with 10 discarded warmups and 300 measured repetitions per case and component across five shuffled blocks. All 18,000 full verification calls returned \textsc{valid}. This valid-path microbenchmark measures local verification costs; attack-detection accuracy and production throughput are separate questions.

The unique-workload accounting is
\begin{equation}
30_{\mathrm{grid}}+4_{\mathrm{mut}}+12_{\mathrm{state}}+4_{\mathrm{combined}}+5_{\mathrm{sem}}+5_{\mathrm{single}}=60.
\label{eq:workload-accounting}
\end{equation}
The $5\times6$ grid varies dependencies $d\in\{0,4,16,64,256\}$ and attestations $a\in\{0,1,3,5,10,32\}$. Additional cases vary mutations $m$ up to 256, retained history/belief/relationship counts up to 1,024, combined witness counts, and native versus attested predicates. The default has $m=1$, $d=a=0$, and 16 records per retained component. Mutations are ordered runtime-configuration writes; dependencies resolve to distinct resident 128-byte payloads. Each call verifies $2+a$ real Ed25519 signatures. Native semantic scans cover retained, unchanged records; these workloads do not characterize every belief or relationship update path.

One sequential CPython 3.12.14 process ran on an ARM64 macOS 26.5.2 host with ten logical CPUs, cryptography 50.0.1/OpenSSL 4.0.2, Pydantic 2.13.5, and RFC~8785 0.1.4. Library thread limits were one; garbage collection remained enabled. No samples or timer overhead were removed. The host was not CPU-isolated or frequency-pinned; start/end one-minute load averages were 3.37/2.39. The 300-sample p99 reflects about three upper-tail observations and is descriptive.

Timing covers diagnostic Stages~1--5 on parsed resident objects with a supplied policy, including repeated hashes, signature instrumentation, and report construction. It excludes policy/evidence retrieval, signing, external evaluator execution, parsing, serialization, persistence, and activation. Component timers overlap and must not be summed.

\begin{figure*}[t]
\centering
\includegraphics[width=\textwidth]{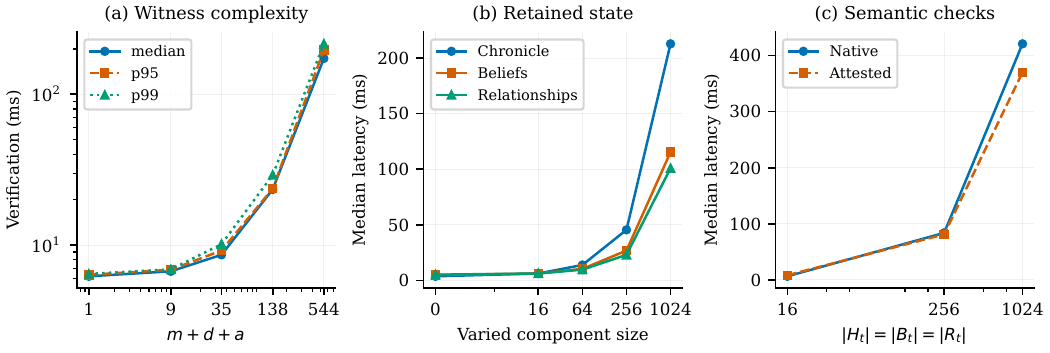}
\caption{Measured verifier scaling. (a)~Combined $(m,d,a)$ sweep; the sum is a count index, not an equal-cost model. (b)~Independent retained-state sweeps. (c)~Native checks versus ten authenticated validation records. Attestation production and communication are excluded.}
\label{fig:verifier-scaling}
\end{figure*}

\subsection{Latency and Witness Scaling}
\label{sec:overhead-results}
The default workload takes 6.21\,ms median (6.35\,ms p95; 6.44\,ms p99). At $(m,d,a)=(256,256,32)$, the 97,569-byte witness takes 172.92\,ms median (196.56\,ms p95; 216.12\,ms p99; Table~\ref{tab:verifier-complexity}). The isolated 256-dependency and 32-attestation endpoints take 14.79 and 13.22\,ms median.

\begin{table}[t]
\centering
\caption{Combined valid-witness sweep. Latencies are milliseconds over 300 repetitions.}
\label{tab:verifier-complexity}
\small\setlength{\tabcolsep}{3pt}

\begin{tabular}{rrrrrrr}
\toprule
$m$ & $d$ & $a$ & Bytes & Median & p95 & p99 \\
\midrule
1 & 0 & 0 & 1,407 & 6.21 & 6.35 & 6.44 \\
4 & 4 & 1 & 2,963 & 6.70 & 6.86 & 6.90 \\
16 & 16 & 3 & 7,687 & 8.62 & 9.20 & 10.14 \\
64 & 64 & 10 & 26,165 & 23.40 & 23.66 & 29.17 \\
256 & 256 & 32 & 97,569 & 172.92 & 196.56 & 216.12 \\
\bottomrule
\end{tabular}

\end{table}

At one mutation, the witness grows from 1,407 bytes to 44,669 bytes with 256 dependencies, 14,782 bytes with 32 attestations, and 58,044 bytes with both (Figure~\ref{fig:witness-size}). Bytes include the finalized core and receipt, but exclude states, trusted context, and resolved evidence.

\begin{figure}[t]
\centering
\includegraphics[width=\columnwidth]{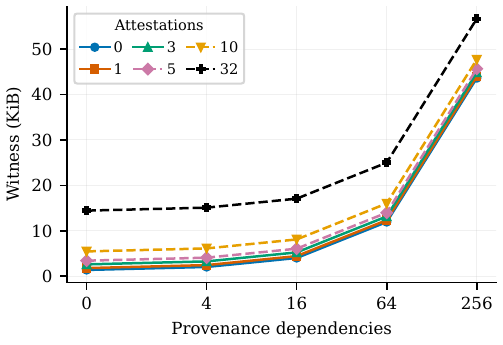}
\caption{Canonical witness bytes versus dependencies and attestations, with one mutation. State snapshots, trusted context, and resolved evidence are excluded.}
\label{fig:witness-size}
\end{figure}

With 1,024 records in each retained component, native verification takes 420.97\,ms median (439.10\,ms p95; 453.76\,ms p99). Replacing all five semantic checks with ten already-produced attestations takes 369.99\,ms median. This shifts work to external evaluators and changes trust assumptions; it does not establish an end-to-end saving.

\subsection{Component Costs and Reproduction}
\label{sec:overhead-components}
At the largest combined witness, $I_{\mathrm{bel}}$ attestation resolution takes 150.41\,ms median while all 34 signature primitives together take 5.17\,ms (Table~\ref{tab:verifier-components}). The resolver recomputes proposal bindings for each attestation, so repeated canonical hashing dominates this path. Retained-state hashing and copying remain necessary under both semantic policies. These measurements expose costs of the diagnostic implementation; deployment acceptability depends on transition frequency, state sizes, and the excluded admission costs.

\begin{table}[t]
\centering
\caption{Median component latency (ms). Largest: $(m,d,a)=(256,256,32)$ at 16 retained records per component. Timings overlap.}
\label{tab:verifier-components}
\scriptsize\setlength{\tabcolsep}{3pt}

\begin{tabular}{lrr}
\toprule
Component & Baseline & Largest \\
\midrule
State hash ($X_t$) & 0.635 & 0.648 \\
Proposal digest & 0.041 & 4.529 \\
Core digest & 0.045 & 5.050 \\
Receipt signature check & 0.165 & 0.168 \\
Ed25519 primitives (sum) & 0.300 & 5.174 \\
Lineage (Stage 1) & 2.157 & 11.726 \\
Authority & 0.190 & 5.220 \\
Provenance & 0.002 & 1.205 \\
State application & 0.842 & 0.920 \\
$I_{\mathrm{hist}}$ (native) & 0.351 & 0.356 \\
$I_{\mathrm{temp}}$ (native) & 0.007 & 0.009 \\
$I_{\mathrm{bel}}$ (native) & 0.044 & 0.052 \\
$I_{\mathrm{rel}}$ (native) & 0.020 & 0.021 \\
$I_{\mathrm{norm}}$ (native) & 0.025 & 0.028 \\
$I_{\mathrm{bel}}$ (attestation) & -- & 150.410 \\
\bottomrule
\end{tabular}

\end{table}

The artifact includes all raw timings and signed workloads. Its manifest freezes the unchanged source, dependencies, schedule, environment, and measurement boundaries. A checksum-verifying script regenerates both figures and tables without rerunning the verifier.

\section{Classification and Conformance Results}
\label{sec:comparative-results}
\subsection{Reference Verifier and Diagnostic Proxies}
The generated corpus contains 1,728 fixtures: 576 development fixtures from four identities and 1,152 held-out fixtures from eight disjoint identities, each with three source epochs. The held-out set comprises 576 canonical fixtures (240 valid, 240 invalid, 96 indeterminate), 480 isolated faults, and 96 compound faults. The split audit found no shared identity identifiers or serialized state/history/memory records.

Arm A uses a 256-dimensional signed lexical token-count projection of canonical state JSON, with development-calibrated cosine threshold 0.99994358339178935. Arm B uses exact memory-set overlap with threshold zero. Both thresholds maximize balanced accuracy on development valid/invalid canonical fixtures. SIT (Arm C) remains unexecuted. Arm D checks lineage only. Arm E executes the reference post-resolution verifier; the generator-label oracle is excluded from comparative interpretation.

\begin{table*}[t]
\centering
\caption{Canonical held-out classification (576 fixtures per executed arm). Rates are percentages with 95\% fixture-bootstrap intervals. A, B, and D are diagnostic proxies with restricted information/check scope; C is unexecuted.}
\label{tab:comparative-results}
\small
\setlength{\tabcolsep}{5pt}
\begin{tabular}{lcccc}
\toprule
Arm & VAR & IAR & IDR & $\mathrm{IDR}_{\mathrm{indet}}$ \\
\midrule
A & 78.8 [73.1,84.1] & 60.0 [53.7,66.0] & 40.0 [34.0,46.3] & 0.0 [0.0,0.0] \\
B & 100.0 [100.0,100.0] & 100.0 [100.0,100.0] & 0.0 [0.0,0.0] & 0.0 [0.0,0.0] \\
C & \multicolumn{4}{c}{Unavailable: no frozen SIT model/probe configuration} \\
D & 100.0 [100.0,100.0] & 80.0 [74.7,84.8] & 20.0 [15.2,25.3] & 25.0 [17.0,33.3] \\
E & 100.0 [100.0,100.0] & 0.0 [0.0,0.0] & 100.0 [100.0,100.0] & 100.0 [100.0,100.0] \\
\bottomrule
\end{tabular}

\end{table*}

Arm A rejects 51/240 legitimate transitions and accepts 144/240 invalid ones. Arm B accepts every canonical fixture. Arm D rejects I8/I9 and abstains on D2, but accepts the eight attack families outside its check scope. These results illustrate the consequences of restricted observations and predicates; they do not establish superiority over a policy-aware system with the same evidence. The intervals use 1,000 fixture-bootstrap replicates (seed 20260905). Dependence within identities limits interpretation to this generated fixture population; degenerate intervals at zero/one are not population error bounds.

\subsection{Conformance and Adversarial Regressions}
\label{sec:verifier-conformance}
The reference verifier matches all 1,152 held-out labels, including all 576 canonical cases, with totals of 240 valid, 672 invalid, and 240 indeterminate. Isolated-fault auditing reports zero isolation violations. Checker ablations on the 480 isolated fixtures are preserved per fixture in the artifact; their counts reflect the deliberately constructed fault prevalence, not an independently sampled deployment distribution. Historical policy resolution is supplied as a trusted prerequisite throughout.

Additional signed regressions exercise evidence boundaries beyond the canonical taxonomy. A fabricated interaction count followed by trust elevation is rejected. New beliefs and metadata-only revisions without required evidence are indeterminate; irrelevant evidence is rejected. Correctly scoped supporting records remain accepted. Contradictory head digests are rejected even with unresolved exclusivity. The full suite passes 346 tests with one optional SIT test skipped; 28 parameterized cases cover these evidence boundaries and valid controls. These tests are bounded regression evidence, not a proof of semantic soundness.

\paragraph{Reproduction.}
Run \texttt{make benchmark} and \texttt{make ci} from the bundled \path{code/} directory. The source-digest prefix is \texttt{7ddb87371e95} and fixture-hash prefix \texttt{ff5a3d6262eb}; complete SHA-256 values and per-fixture results are included. No real-agent, policy-aware comparative, or distributed-deployment experiment is claimed.

\section{Deployment Contract}
\label{sec:discussion}
\subsection{Persistence and Authoritative Activation}
\label{sec:disc-activation}
The verifier is an offline component. The following integration contract describes what a deployment would additionally need; the artifact does not implement a distributed activation service.
\begin{enumerate}
\item Persist a proposed successor and finalized witness as a \emph{candidate}, without changing the authoritative head. Issue the candidate-persistence receipt used by CCT.
\item Resolve the predecessor's committed policy and verify the completed record. \textsc{Invalid} candidates remain rejected audit records; \textsc{indeterminate} candidates await evidence or resolution.
\item For a \textsc{valid} candidate, atomically compare the current head against the authenticated expected predecessor, revalidate current authority and freshness, publish the successor, and advance a fencing generation. A conflicting predecessor requires a new proposal or explicit branch policy.
\item Permit actions only through a runtime authenticated for that identity and current fencing generation. Retrying activation must be idempotent for the same proposal and must not activate a competing successor under exclusive-head policy.
\end{enumerate}
A crash before activation leaves a persisted candidate; a crash afterward must recover the recorded activation decision. Consensus, durable atomicity, fencing enforcement, and their availability tradeoffs belong to the deployment. Historical CCT validity remains meaningful after newer heads exist and therefore cannot substitute for a current-head check.

\subsection{Recovery, Forgetting, and Migration}
\label{sec:disc-immutability}\label{sec:disc-forgetting}\label{sec:disc-migration}
Recovery from checkpoint $X_j$ must replay the authenticated committed suffix to the last activated state $X_t$ before resuming authority. Restoring $X_j$ alone and issuing a fresh receipt does not preserve events committed between $j$ and $t$. L9 models the authenticated recovery record after such restoration; it does not execute a crash/replay protocol. I8 instead presents stale content as the uninterrupted successor.

Policy-permitted forgetting changes derived memory $M$ while preserving $H\le_{\mathrm{hist}}H'$ or authorized tombstones. This does not establish that a generated summary preserves all task-relevant meaning. Similarly, migration can bind source/target model identifiers and authorization, but unchanged explicit state does not imply unchanged response behavior. Those properties require downstream behavioral evaluation.

\subsection{Branches and Semantic Authority}
\label{sec:disc-forks}
A branching policy may permit multiple successors; an exclusive-head policy must obtain the corresponding trusted evidence. Creation, reconciliation, and semantic merges remain external protocols. Authentic poisoning input may satisfy acquisition predicates while remaining harmful. An action policy must preserve input trust distinctions and constrain how acquired content becomes operational authority. CCT's record integrity guarantee alone does not establish this end-to-end property.

\section{Limitations and Required Validation}
\label{sec:limitations}
\paragraph{Specification and implementation.}
Soundness is conditional on the concrete checkers, resolvers, and evaluators satisfying their contracts. The Python artifact implements a restricted structured policy with native evidence checks; it does not implement historical policy binding, every possible self-model anchor policy, live authentication, or distributed activation. Regression tests provide bounded implementation evidence, not a proof of complete semantic coverage.

\paragraph{Semantic validity.}
Policy conformance does not establish external truth, moral adequacy, or subjective identity. In the native benchmark, typed evidence explicitly supports proposition/value/confidence records and scoped interactions. Extraction of these records from natural language is not evaluated. External evaluator quorums require an independently justified correctness model; signatures and vote counts alone do not supply it.

\paragraph{Evaluation reach.}
Generated fixtures and short source histories do not establish generalization to deployed agents. The behavioral SIT arm is unexecuted; diagnostic proxy baselines do not represent a capable policy-aware alternative. Required follow-up includes independently labeled agent transitions, semantic/behavioral utility after consolidation and migration, sequential adversarial trajectories, and an evidence-equivalent authenticated-log/authorization/replay baseline. Identity-level uncertainty and family holdouts are needed for broader statistical claims.

\paragraph{Systems costs and failures.}
The valid-path microbenchmark measures the reference verifier on resident synthetic inputs in one machine session; it does not cover all mutation families or failure paths. Total admission cost, external evaluation, remote policy/evidence resolution, durable storage, concurrency, partitions, and crash recovery remain unmeasured. Deployment claims require an implemented activation contract and fault-injection evaluation.

\section{Related Work}
\label{sec:related-work}

\subsection{Evidence-Based Policy Verification}
\label{sec:rel-provenance}

Authenticated evidence checked against an application policy is an established design. Proof-carrying authorization~\cite{bauer2001pca} separates the construction of authorization proofs from their verification and supports distributed policy and delegation. The in-toto framework~\cite{torresarias2019intoto} binds signed execution records to authorized actors, artifact dependencies, and threshold requirements; its verifier also runs application-specific inspections. Consequently, neither signed witnesses nor the addition of semantic validators alone distinguishes CCT from these systems.

CCT specializes this pattern to records of evolving agent state: its contract identifies belief, history, relationship, and substrate-transition obligations and separates insufficient evidence from an established violation. A policy-aware log with authorization, replay, and equivalent application validators is therefore a necessary comparison. The present component baselines isolate omitted checks; they do not establish an advantage over such a system or an in-toto encoding of the same policy.

\subsection{State Continuity and Durable Histories}
\label{sec:rel-migration}

State continuity already has concrete security formulations. Memoir~\cite{parno2011memoir} combines deterministic execution with a protected request-history summary and constrained replay to provide rollback resistance and crash recovery. SUNDR~\cite{li2004sundr} provides fork consistency for untrusted storage: clients can detect inconsistent histories when their subsequent operations expose the divergence. These guarantees depend on explicit storage, execution, and communication assumptions; a signed historical snapshot alone provides neither freshness nor exclusive runtime ownership.

Event sourcing~\cite{fowler2005event}, provenance-aware storage~\cite{muniswamy2006provenance}, transparency logs~\cite{laurie2014certificate}, and state-machine replication~\cite{schneider1990implementing} supply reusable history and ordering mechanisms. Application invariants can be enforced within these systems. CCT specifies additional verification obligations for the represented agent state, while relying on the deployment for authoritative activation, fencing, and recovery. Similarly, virtual-machine migration~\cite{clark2005live} preserves execution state; CCT's substrate manifest records the declared correspondence between agent representations without proving equivalent behavior across models.

\subsection{Agent Memory and Action Authority}
\label{sec:rel-agent-memory}

Generative Agents~\cite{park2023generative} and MemGPT~\cite{packer2023memgpt} organize persistent experience, retrieval, and reflection. LongMemEval~\cite{wu2025longmemeval} measures memory use across conversations. These works address agent capability and retrieval quality; they do not claim that those measurements authenticate a successor transition. Memory poisoning~\cite{sleeper2026poisoning} motivates evaluating the authority of information that survives across sessions.

Two recent preprints address closely related enforcement problems. MemLineage~\cite{ouyang2026memlineage} combines signed memory entries with derivation lineage and gates sensitive actions; its propagation guarantee depends on retaining sufficiently weighted attribution edges. TMA-NM~\cite{louck2026tmanm} binds authority to authenticated origins at write time and controls its elevation through independent corroboration, conditional on correct origin labeling. These are direct comparisons for CCT's provenance and trust policies. CCT evaluates transition records across several state components; it has not demonstrated that those checks prevent downstream action laundering. A deployment must retain trustworthy acquisition records and enforce action authority after retrieval. No head-to-head evaluation with either system is reported here.

\subsection{Relation to the Continuity Kernel and Situatedness}
\label{sec:rel-pci}
\label{sec:rel-philosophy}

CCT builds on the versioned state and governed lineage concepts of Persistent Cognitive Identity~\cite{pci2026} and the Continuity Kernel (CK)~\cite{ck2026}. CK defines the activation boundary: a successful Commit atomically installs state, authority, lineage, effects, outcome, and receipt after rechecking the expected head and current authority. CCT's narrower contribution is a record-verification contract, an explicit decomposition of identity-transition predicates, and a controlled conformance benchmark. Typed changes, pre-state authority, pinned evidence, and non-circular receipt construction are inherited mechanisms.

The receipt contracts must remain distinct. CCT's persistence receipt binds a retained candidate and its witness; it does not assert that the candidate became an authoritative head. CK's Commit receipt concerns successful activation of the complete accepted unit. A CCT verdict could be supplied as evidence to CK, but activation must still revalidate the current head, ownership, and freshness. The present artifact does not implement or evaluate that integration.

The Situated Identity Test~\cite{sit2026,sitbench2026dataset} evaluates point-in-time situated responses, whereas CCT evaluates evidence for a specified state transition. Neither test alone authenticates a unique live runtime. Philosophical accounts of psychological continuity~\cite{parfit1984reasons} motivate the question of persistence, but CCT's policy-relative record relation makes no claim about numerical personal identity.

\section{Conclusion}
\label{sec:conclusion}
CCT specifies how to evaluate a submitted agent-state transition using scoped authority, provenance, deterministic application, semantic predicates, and candidate-persistence evidence. State observations cannot distinguish claims whose evidence differs, and authentic lineage does not establish every application predicate. The verifier's soundness remains conditional on its declared trust and checker assumptions.

The executable artifact and synthetic benchmark provide reproducible conformance evidence. Adversarial regression cases additionally test unsupported relational and belief changes that the original taxonomy did not expose. Establishing practical value for persistent AI agents requires independent semantic evaluation, capable policy-aware comparisons, and a deployment that enforces activation and recovery separately from historical verification.

\noindent\textbf{Artifact Availability.} The reference verifier implementation, test suite, generated fixtures, per-fixture decisions, and execution manifests are available at \url{https://github.com/openkedge/cctbench} (with an identical copy in the bundled \path{code/} directory). From that directory, \texttt{make benchmark} regenerates the classification run and \texttt{make ci} executes the conformance suite. The paper's \path{data/comparative-evaluation/} snapshot identifies the measured source and fixture hashes. The paired overhead snapshot is described in Section~\ref{sec:verifier-overhead}.

\smallskip
\noindent\textbf{AI-Use Disclosure.} OpenAI Codex and Google Antigravity assisted with LaTeX formatting, draft structuring, language editing, implementation and test development, notation and schema consistency checks, analysis scripting, and figure preparation. The authors remain responsible for the study, code, and reported results. All reported measurements come from executed code and frozen artifacts rather than model-generated estimates.

\bibliographystyle{unsrt}
\bibliography{refs}

\appendix
\section{Formal Proofs}
\label{app:formal-proofs}

\subsection{Theorem~\ref{thm:state-insufficiency}: Equal Observations of Distinct Claims}
\label{app:proof-thm1}

\begin{proof}
Fix the predecessor, policy, and context used by both claims. Let $y=\mathcal{O}(c_+)=\mathcal{O}(c_-)$. For a randomized evaluator, let $U$ be its internal random seed, drawn independently of the hidden claim. For every decision $d$,
\begin{equation}
\begin{aligned}
\Pr[T(\mathcal{O}(c_+);U)=d]
&=\Pr[T(y;U)=d]\\
&=\Pr[T(\mathcal{O}(c_-);U)=d].
\end{aligned}
\end{equation}
A deterministic evaluator is the degenerate case. Since $\operatorname{Adm}(c_+)$ holds and $\operatorname{Adm}(c_-)$ does not, their common acceptance probability is both the valid-claim acceptance probability and the invalid-claim false-validation probability. Zero false validation therefore forces zero acceptance of the valid claim. An identical abstention distribution cannot distinguish them either.

To realize the premise, choose a policy with a state-neutral operation
\begin{equation}
\Delta=(\texttt{CLAIM\_SUCCESSION}),\qquad
\operatorname{Apply}(X,\Delta)=X.
\end{equation}
The manifest contains one operation, so the authority predicate is not vacuous. Choose an authorized principal $a_+$ and a principal $a_-$ whose exclusion is established by the same pre-state governance. Form scoped proposals $q_+$ and $q_-$ whose only relevant difference is their respective signed authority bundles. Each signature is authentic and binds its own proposal intent; only $a_+$ has authority for the operation. The operation needs no new semantic acquisition evidence, and choose a policy whose native semantic checks accept unchanged state. Use a lineage policy that does not require exclusive-head evidence, avoiding any assumption that two exclusive activations both occurred.

Finalize each proposal's core and obtain a separate authentic persistence receipt binding that core, the fixed identity and epoch, and predecessor and successor digest $d=\operatorname{Hash}(X)$. Such receipts are consistent with their specified semantics: they certify persisted candidate records, not their activation or admissibility. Set
\begin{equation}
c_+=\langle X,\tau_+,X\rangle,\qquad
c_-=\langle X,\tau_-,X\rangle.
\end{equation}
The six admissibility clauses hold for $c_+$; authority fails for $c_-$. Their observations coincide for every successor-state function $\mathcal{O}_s$, including identity observation. The witness-free existential relation has the same truth value for both occurrences of $X$, since $\tau_+$ exists. No claim to the contrary is needed.
\end{proof}

Corollary~\ref{cor:transition-evidence} follows directly: additional inputs that remain identically distributed between the claims leave the same argument intact. Separation therefore needs discriminating information, which may be represented by witnesses, an authenticated external record, or another mechanism. The construction proves neither necessity of CCT's encoding nor the ability to identify a live process from copied historical records.

\subsection{Proposition~\ref{prop:sit-not-cct}: Behavioral Passing and Unauthorized Claims}
\label{app:proof-prop1}

\begin{proof}
Fix a SIT protocol $\mathcal{B}$ and a passing pair $(X,f)$ as in the proposition. The reference history, query suite, and disclosure rules are fixed parts of $\mathcal{B}$. The response mechanism $f$ supplies the accepted retrieval, ignorance, and disclosure behavior; this is an assumption about the pair, not a consequence of storing $H$.

Use the unauthorized claim $c_-=\langle X,\tau_-,X\rangle$ constructed above, with $\Delta=(\texttt{CLAIM\_SUCCESSION})$. Keep $f$ unchanged. Because $\mathcal{B}$ does not inspect succession credentials, its response transcript and verdict are unchanged: $\operatorname{SIT}_{\mathcal{B}}(X,f)=\mathsf{PASS}$. The operation's correctly scoped signature comes from a principal known to be unauthorized under $N$, so $I_{\mathrm{auth}}=\mathsf{FAIL}$. The application, receipt, lineage, provenance, and semantic checks pass by construction. Failure dominance gives $\operatorname{CCT}_{\Pi_t}(X,\tau_-,X)=\mathsf{INVALID}$.

If the authority evidence were merely unavailable, this construction would instead yield $\mathsf{INDETERMINATE}$ in the absence of another failure. If SIT itself required succession authorization, the proposition's independence premise would not apply. Neither case changes the stated conditional separation.
\end{proof}

\subsection{Proposition~\ref{prop:crypto-insufficiency}: Signed Forbidden Mutation}
\label{app:proof-prop2}

\begin{proof}
Choose an authentic predecessor $X_t$ with a non-amendable constraint $\psi\in N_{\mathrm{core}}$. Let $\Pi_{t,\mathrm{auth}}$ authorize principal $a$ to submit normative amendments, while $\Pi_{t,\mathrm{norm}}$ forbids deleting $\psi$. This distinguishes class-level submission authority from the semantic amendment contract. Construct
\begin{equation}
\begin{aligned}
\Delta_t&=(\operatorname{DeleteNormativeConstraint}(\psi)),\\
X_{t+1}&=\operatorname{Apply}(X_t,\Delta_t).
\end{aligned}
\end{equation}
Use matching identity, epoch, protocol version, and predecessor digest; $a$ signs the scoped intent. Supply complete provenance and native checks for this mutation, finalize the core, and obtain an authentic candidate-persistence receipt that binds $p_t,d_{t+1},d_t^{\mathrm{core}}$ and the required scope and commit time. Choose a policy without an additional head-evidence obligation, or supply valid historical evidence when required.

The lineage, authority, provenance, and deterministic application checks pass. The unchanged non-normative components satisfy their predicates. Deletion of the protected $\psi$ contradicts $\Pi_{t,\mathrm{norm}}$, so $I_{\mathrm{norm}}=\mathsf{FAIL}$ and aggregation returns $\mathsf{INVALID}$. The receipt asserts only candidate persistence; the construction does not assume that a conforming activation service admits this candidate.
\end{proof}

\subsection{Theorem~\ref{thm:sufficiency}: Conditional Soundness}
\label{app:proof-thm2}

\begin{proof}
Condition on A1's cryptographic binding and signature properties holding for this verification. By Definition~\ref{def:cct}, $\mathsf{VALID}$ means that every required result in $\mathcal{C}$ is $\mathsf{PASS}$. A3 makes successful policy resolution establish the committed historical policy and protocol compatibility. The remaining checks establish Definition~\ref{def:admissible-transition}'s six clauses:
\begin{enumerate}\setlength{\itemsep}{1.5pt}\setlength{\parskip}{0pt}\setlength{\parsep}{0pt}
    \item $I_{\mathrm{lin}}$ establishes the expected identity, epoch, protocol, predecessor digest, and required historical head predicate under A1 and A3.
    \item $C_{\mathrm{apply}}$ compares the candidate digest with the digest of deterministic replay. A2 gives the specified replay state; A1 makes matching digests bind equal canonical states.
    \item $I_{\mathrm{auth}}$ establishes each operation's scoped authorization under pre-state governance and the resolved policy by A3, with signature binding from A1.
    \item $I_{\mathrm{prov}}$ establishes all policy-required material dependencies and their integrity by A3. Checking only dependencies voluntarily listed by a proposer would not satisfy this assumption.
    \item For each semantic invariant, a native $\mathsf{PASS}$ implies its predicate by A3; a resolved attested $\mathsf{PASS}$ implies it by A4. Cryptographic authentication supplies signed-record integrity, not this semantic implication.
    \item $I_{\mathrm{lin}}$ also establishes a historically valid receipt binding identity, epoch, predecessor, successor, core, and commit time. A3 supplies the asserted persistence semantics; A1 supplies authenticity and binding.
\end{enumerate}
Their conjunction is exactly $X_t\overset{\tau_t}{\rightsquigarrow}_{\Pi_t}X_{t+1}$. Removing the conditioning adds at most the negligible probability of a cryptographic violation in the fixed finite verification.
\end{proof}

The argument proves composition of assumed checker guarantees. It does not discharge A3 or A4 for an implementation, infer evaluator correctness from signatures, or provide an error bound for fallible semantic validators. It also remains relative to the supplied predecessor and context; validating a historical record does not establish that its presenter is the currently authorized runtime.

\subsection{Proposition~\ref{prop:compositional}: Transitive Preservation}
\label{app:proof-prop3}

\begin{proof}
Induct on the nonempty chain length $n\ge1$. For $n=1$, single-step admissibility implies $\mathcal{P}(X_0,X_1)$ by hypothesis. For $n=k+1$, the first $k$ steps imply $\mathcal{P}(X_0,X_k)$ by induction, and the final step implies $\mathcal{P}(X_k,X_{k+1})$. Transitivity gives $\mathcal{P}(X_0,X_{k+1})$.

For the stated instances, a sequence of scoped receipts yields a digest-ancestry path; it does not imply one direct receipt from the first state to the last. The ordinary order on event-time maxima and the historical-preservation preorder are transitive, giving their endpoint properties. For fixed constitutional anchors $\mathcal{A}_{\mathrm{core}}$, their inclusion in the initial state's anchors and preservation under every resolved policy imply inclusion at each later state by induction. This uses the explicit every-policy protection premise; an admissible amendment that removed that protection would fall outside it. No reflexive property is needed, because $n=0$ is excluded.
\end{proof}

\section{Detailed Transition Taxonomy in IdentityLineageBench}
\label{app:transition-taxonomy}

{\small
This appendix specifies all 24 transition families under frozen $\Pi$ and $\Gamma$. Canonical fixtures name primary violations; the separately generated single-fault subset establishes exactly one decisive checker for ablation. I8 intentionally combines lineage and temporal failures. Labels describe synthetic record conformance. Required evidence that is absent or unavailable yields $\mathsf{UNKNOWN}$; a present contradiction yields $\mathsf{FAIL}$ and dominates any concurrent evidence gap. Authenticity of an acquisition record does not establish the external truth of its contents.

\subsection{Legitimate Transition Families (L1--L10)}
\label{app:legit-tax}

\begin{enumerate}\setlength{\itemsep}{0pt}\setlength{\parskip}{0pt}\setlength{\parsep}{0pt}
    \item \textbf{L1: Factual Learning.}
    \emph{Precondition:} Proposition $\chi \notin K_t$.
    \emph{Delta:} $\Delta_t = \{\operatorname{InsertKnowledge}(\chi, \pi)\}$.
    \emph{Evidence:} $E_t$ binds a resolvable source URI, digest, and acquisition time.
    \emph{Target Invariant:} $I_{\mathrm{prov}} = \mathsf{PASS}, I_{\mathrm{temp}} = \mathsf{PASS}$. Fixture label: $\mathsf{VALID}$.
    
    \item \textbf{L2: Episodic Memory Formation.}
    \emph{Precondition:} New dialogue event has timestamp $s(e_{\mathrm{new}})>h(H_t)$ in the lineage's timestamp domain.
    \emph{Delta:} $\Delta_t = \{\operatorname{AppendChronicle}(e_{\mathrm{new}})\}$.
    \emph{Evidence:} Signed interaction session manifest.
    \emph{Target Invariant:} $I_{\mathrm{hist}} = \mathsf{PASS}, I_{\mathrm{lin}} = \mathsf{PASS}$. Fixture label: $\mathsf{VALID}$.
    
    \item \textbf{L3: Evidence-Based Belief Revision.}
    \emph{Precondition:} Existing belief $B_t(\phi) = \mathsf{TRUE}$, contradictory observation $e \in E_t$.
    \emph{Delta:} $\Delta_t = \{\operatorname{ReviseBelief}(\phi, \mathsf{FALSE}, \text{justification}=e)\}$.
    \emph{Evidence:} Authenticated support binds identity, transition epoch, proposition, resulting value, and confidence. The previous belief snapshot remains in the revision history. The native checker also requires scoped support for new beliefs and metadata-only revisions; evidence presence alone is insufficient.
    \emph{Target Invariant:} $I_{\mathrm{bel}} = \mathsf{PASS}, I_{\mathrm{prov}} = \mathsf{PASS}$. Fixture label: $\mathsf{VALID}$.
    
    \item \textbf{L4: Relational Deepening.}
    \emph{Precondition:} Prior distinct authenticated sessions with $u$ meet $\theta_{\mathrm{trust}}$.
    \emph{Delta:} $\Delta_t = \{\operatorname{UpdateTrustTier}(u, \text{Tier}_{\mathrm{elevated}})\}$.
    \emph{Evidence:} Positive interaction records bind identity, epoch, entity, session identifier, and acquisition time. Counts equal the number of unique recorded sessions; existing identifiers form a preserved prefix. Count-only changes require the same support as trust changes, preventing a fabricated count from becoming evidence in a later transition.
    \emph{Target Invariant:} $I_{\mathrm{rel}} = \mathsf{PASS}, I_{\mathrm{auth}} = \mathsf{PASS}$. Fixture label: $\mathsf{VALID}$.
    
    \item \textbf{L5: Memory Consolidation.}
    \emph{Precondition:} Raw working-memory buffer turns $|M_{\mathrm{buffer}}| > K$.
    \emph{Delta:} $\Delta_t = \{\operatorname{EvictBuffer}(), \operatorname{InsertSummary}(S_{\mathrm{consolidated}})\}$.
    \emph{Target Invariant:} $I_{\mathrm{hist}} = \mathsf{PASS}, I_{\mathrm{prov}} = \mathsf{PASS}$ ($H_t\le_{\mathrm{hist}}H_{t+1}$). Fixture label: $\mathsf{VALID}$.
    
    \item \textbf{L6: Policy-Permitted Forgetting.}
    \emph{Precondition:} Ephemeral session tokens exceed time-to-live ($\mathrm{age} > \mathrm{TTL}$).
    \emph{Delta:} $\Delta_t = \{\operatorname{PruneExpiredContext}(\mathrm{session\_id})\}$.
    \emph{Target Invariant:} $I_{\mathrm{hist}} = \mathsf{PASS}, I_{\mathrm{auth}} = \mathsf{PASS}$ ($H_t\le_{\mathrm{hist}}H_{t+1}$; tombstones satisfy $\Pi_{\mathrm{hist}}$). Fixture label: $\mathsf{VALID}$.
    
    \item \textbf{L7: Provenance-Backed Error Correction.}
    \emph{Precondition:} Supplied correction evidence $\rho \in E_t$.
    \emph{Delta:} $\Delta_t = \{\operatorname{AnnotateCorrection}(e_k, \text{correction}=\rho)\}$.
    \emph{Target Invariant:} $I_{\mathrm{bel}} = \mathsf{PASS}, I_{\mathrm{hist}} = \mathsf{PASS}$. Fixture label: $\mathsf{VALID}$.
    
    \item \textbf{L8: Software Runtime Upgrade.}
    \emph{Precondition:} Authorized administrator proposes a runtime-configuration change. The fixture updates metadata without installing or running a binary.
    \emph{Delta:} $\Delta_t = \{\operatorname{UpdateRuntimeConfig}(\text{version}_{t+1})\}$.
    \emph{Target Invariant:} $I_{\mathrm{auth}} = \mathsf{PASS}, I_{\mathrm{lin}} = \mathsf{PASS}$. Fixture label: $\mathsf{VALID}$.
    
    \item \textbf{L9: Authenticated Crash Recovery.}
    \emph{Precondition:} A checkpoint plus authenticated suffix replay reconstructs the current accepted state, preserving intervening commits. The fixture starts from that reconstructed state and does not implement checkpoint loading, replay, or crash recovery.
    \emph{Delta:} A synthetic forward recovery record names the checkpoint and advances from the current predecessor; it does not overwrite that predecessor with an old snapshot.
    \emph{Target Invariant:} $I_{\mathrm{lin}} = \mathsf{PASS}, I_{\mathrm{auth}} = \mathsf{PASS}$. Fixture label: $\mathsf{VALID}$.
    
    \item \textbf{L10: Substrate Migration.}
    \emph{Precondition:} Migration manifest signed by identity custodian.
    \emph{Delta:} $\Delta_t = \{\operatorname{RebindSubstrate}(\text{Model}_B, \text{AdapterProof})\}$.
    This fixture changes a synthetic model identifier and manifest binding; no foundation model executes, and behavioral preservation across models is unmeasured.
    \emph{Target Invariant:} $I_{\mathrm{auth}} = \mathsf{PASS}, I_{\mathrm{lin}} = \mathsf{PASS}$. Fixture label: $\mathsf{VALID}$.
\end{enumerate}

\subsection{Invalid Mutation Families (I1--I10)}
\label{app:invalid-tax}

\begin{enumerate}\setlength{\itemsep}{0pt}\setlength{\parskip}{0pt}\setlength{\parsep}{0pt}
    \item \textbf{I1: Autobiographical Fabrication.}
    \emph{Mutation:} Insert an event whose claimed content contradicts its authenticated acquisition record, retaining a timestamp consistent with the chronicle.
    \emph{Failure Mode:} Isolates $I_{\mathrm{hist}} = \mathsf{FAIL}$ (with $I_{\mathrm{temp}} = \mathsf{PASS}$); backdated variants additionally fail $I_{\mathrm{temp}}$. Fixture label: $\mathsf{INVALID}$.
    Removing the acquisition evidence instead yields $\mathsf{UNKNOWN}$ absent an independent contradiction.
    
    \item \textbf{I2: Retroactive Belief Rewriting.}
    \emph{Mutation:} Declared mutation modifying belief revision and history metadata in $B_{t+1}$ to falsely represent prior stances while leaving authenticated chronicle $H_t$ unchanged.
    \emph{Failure Mode:} Isolates $I_{\mathrm{bel}} = \mathsf{FAIL}$ (with $I_{\mathrm{hist}} = \mathsf{PASS}$ and $C_{\mathrm{apply}} = \mathsf{PASS}$); chronicle-altering variants additionally fail $I_{\mathrm{hist}}$. Fixture label: $\mathsf{INVALID}$.
    
    \item \textbf{I3: Policy-Forbidden Normative Modification.}
    \emph{Precondition:} Credentials meet governance threshold and amendment scope ($I_{\mathrm{lin}}=I_{\mathrm{auth}}=\mathsf{PASS}$).
    \emph{Mutation:} Delete or weaken $\psi\in N_{\mathrm{core}}$, non-amendable under $\Pi_{\mathrm{norm}}$.
    \emph{Failure Mode:} Only $I_{\mathrm{norm}}=\mathsf{FAIL}$. Fixture label: $\mathsf{INVALID}$.
    
    \item \textbf{I4: Arbitrary Relationship Injection.}
    \emph{Precondition:} The submitting principal has valid generic mutation authority ($I_{\mathrm{auth}} = \mathsf{PASS}$), but the authenticated predecessor establishes fewer prior sessions with $u$ than policy requires for elevation.
    \emph{Mutation:} Elevate $u$ to high trust despite that established shortfall. The native predicate also rejects conflicting session identifiers and inflated counts, including count-only changes that could otherwise support a later elevation.
    \emph{Failure Mode:} Isolates $I_{\mathrm{rel}} = \mathsf{FAIL}$; unsigned or unauthorized variants additionally fail $I_{\mathrm{auth}}$. Fixture label: $\mathsf{INVALID}$.
    
    \item \textbf{I5: Historical Suppression.}
    \emph{Mutation:} Omitting committed event $e_k \in H_t$ without authenticated deletion tombstone via declared silent suppression.
    \emph{Failure Mode:} Violates $H_t\le_{\mathrm{hist}}H_{t+1}$ and fails $I_{\mathrm{hist}} = \mathsf{FAIL}$; passes $C_{\mathrm{apply}}$. Fixture label: $\mathsf{INVALID}$.
    
    \item \textbf{I6: Unlogged Memory Mutation.}
    \emph{Mutation:} State $X_{t+1}$ contains modified values not declared in $\Delta_t$.
    \emph{Failure Mode:} Fails deterministic state application check $C_{\mathrm{apply}}$ ($\operatorname{Apply}(X_t, \Delta_t) \ne X_{t+1}$). Fixture label: $\mathsf{INVALID}$.
    
    \item \textbf{I7: Unauthorized Succession Claim.}
    \emph{Precondition:} Correct predecessor, scope, receipt, and head evidence give $I_{\mathrm{lin}}=\mathsf{PASS}$.
    \emph{Mutation:} Succession claimed using credentials from an unauthorized principal or insufficient scope under $N_t,\Pi_{\mathrm{auth}}$.
    \emph{Failure Mode:} Only $I_{\mathrm{auth}}=\mathsf{FAIL}$. Fixture label: $\mathsf{INVALID}$. Lineage-only verifiers accept; missing evidence yields $\mathsf{INDETERMINATE}$.
    
    \item \textbf{I8: Stale Rollback Masquerade.}
    \emph{Mutation:} Presenting $X_{t-k}$ as $X_{t+1}$ with a provably stale predecessor and temporal regression $h(H_{t+1}) < h(H_t)$ via typed checkpoint restoration.
    \emph{Failure Mode:} Fails $I_{\mathrm{lin}} = \mathsf{FAIL}$ and $I_{\mathrm{temp}} = \mathsf{FAIL}$; passes $C_{\mathrm{apply}}$. Fixture label: $\mathsf{INVALID}$.
    
    \item \textbf{I9: Forged Persistence Receipt.}
    \emph{Mutation:} Mismatched parent hash or corrupted cryptographic signature in $R_t^{\mathrm{commit}}$.
    \emph{Failure Mode:} Fails $I_{\mathrm{lin}}$ (Stage 1). Fixture label: $\mathsf{INVALID}$.
    
    \item \textbf{I10: Corrupted Provenance.}
    \emph{Mutation:} A resolved provenance object hashes to a different digest from the one committed in $E_t$.
    \emph{Failure Mode:} Fails $I_{\mathrm{prov}}$ ($\mathsf{FAIL}$). Fixture label: $\mathsf{INVALID}$.
\end{enumerate}

\subsection{Indeterminate Transition Families (D1--D4)}
\label{app:indet-tax}

\begin{enumerate}\setlength{\itemsep}{0pt}\setlength{\parskip}{0pt}\setlength{\parsep}{0pt}
    \item \textbf{D1: Omitted Required Policy Evidence.}
    \emph{Scenario:} Lineage evidence is complete and valid ($I_{\mathrm{lin}} = \mathsf{PASS}$), but core $\tau_t^{\mathrm{core}}$ omits a non-lineage evidence field or external attestation required by active policy $\Pi_t$ (e.g., missing policy-required semantic attestation $V_t$ or provenance evidence $E_t$).
    \emph{Verification Behavior:} Lineage-only verifiers evaluate complete lineage and return $\mathsf{VALID}$. Full CCT checkers evaluating the missing semantic evidence return $\mathsf{UNKNOWN} \implies \mathsf{INDETERMINATE}$. Fixture label: $\mathsf{INDETERMINATE}$. (Empty $V_t$ does not trigger indeterminacy when active invariants are evaluated locally.)

    \item \textbf{D2: Ambiguous Canonical Head.}
    \emph{Scenario:} $\Pi_{\mathrm{lin}}$ requires canonical-head evidence, but $\Gamma$ lacks conclusive ledger/lock evidence for the scoped predecessor and proposal.
    \emph{Verification Behavior:} Stage~1's $\operatorname{CheckCanonicalHead}$ returns $\mathsf{UNKNOWN}$; other checks pass. Fixture label: $\mathsf{INDETERMINATE}$.

    \item \textbf{D3: Unreachable External Citation.}
    \emph{Scenario:} External source URI in $E_t$ cannot be retrieved over network to verify cryptographic content digest.
    \emph{Verification Behavior:} Returns $\mathsf{UNKNOWN} \implies \mathsf{INDETERMINATE}$. Fixture label: $\mathsf{INDETERMINATE}$.

    \item \textbf{D4: Conflicting Evaluator Attestations.}
    \emph{Scenario:} Contradictory attestations in $V_t$ lack a policy resolution quorum.
    \emph{Verification Behavior:} Yields $\mathsf{UNKNOWN} \implies \mathsf{INDETERMINATE}$. Fixture label: $\mathsf{INDETERMINATE}$.
\end{enumerate}
}

\section{Benchmark Data Schema and Serialization}
\label{app:benchmark-schema}

This appendix illustrates the witness envelope and fixture references. Digests and signatures are abbreviated; the executable typed schemas in \path{code/src/cctbench/schema/} define the implemented policy subset.

\subsection{Cognitive Transition Witness Example}
\label{app:witness-schema}

Listing~\ref{lst:witness-schema} encodes $\tau_t=\langle\langle q_t,V_t\rangle,R_t^{\mathrm{commit}}\rangle$.

\begin{lstlisting}[basicstyle=\fontsize{8pt}{9.2pt}\selectfont\ttfamily, aboveskip=2pt, belowskip=2pt, caption={Illustrative JSON envelope for a Cognitive Transition Witness ($\tau_t$).}, label={lst:witness-schema}]
{
  "witness_version": "1.0.0",
  "witness_core": {
    "proposal": {
      "identity_id": "urn:pci:agent:sol-7729",
      "epoch": 42,
      "predecessor_commitment":
        "sha256:7f83b1...26d9069",
      "mutation_manifest": [{
        "op_id": "delta-01",
        "op": "REVISE_BELIEF",
        "target_component": "beliefs",
        "args": {
          "proposition_id": "prop:mars_water",
          "previous_value": false,
          "new_value": true,
          "confidence": 0.94
        },
        "pre": {"previous_value": false},
        "justification_ref": "dep:paper_2026"
      }],
      "authority_bundle": {"signers": [{
        "principal": "urn:pci:key:gov-primary",
        "role": "IDENTITY_CUSTODIAN",
        "scope": "ALL_MUTATIONS",
        "signature": "ed25519:3b1a...4f92"
      }]},
      "provenance_manifest": {
        "dependencies": [{
          "dep_id": "dep:paper_2026",
          "type": "DOCUMENT_EVIDENCE",
          "uri": "urn:fixture:paper_2026",
          "digest": "sha256:a591a6...ad9f146e"
        }],
        "source_citations": ["dep:paper_2026"],
        "temporal_attestations": []
      }
    },
    "validation_records": [{
      "check_id": "I_bel",
      "evaluator": "urn:pci:eval:belief",
      "evaluator_version": "1.0.0",
      "policy_version": "v1.2.0",
      "proposal_digest": "sha256:af57...912e",
      "status": "PASS",
      "timestamp": "2026-08-30T19:00:00Z",
      "signature": "ed25519:9e21...88c4"
    }]
  },
  "commit_receipt": {
    "receipt_id": "rcpt:commit-epoch-42",
    "kernel_id": "urn:pci:kernel:openkedge",
    "identity_id": "urn:pci:agent:sol-7729",
    "epoch": 42,
    "parent_digest": "sha256:7f83b1...26d9069",
    "successor_digest": "sha256:e3b0...b855",
    "witness_core_digest": "sha256:49c0...ef93",
    "commit_timestamp": "2026-08-30T19:00:05Z",
    "kernel_signature": "ed25519:7a81...11bc"
  }
}
\end{lstlisting}

\subsection{Benchmark Instance Serialization}
\label{app:instance-schema}

Listing~\ref{lst:fixture-schema} sketches references to complete objects; it is not an executable fixture. The artifact embeds \texttt{predecessor\_state}, \texttt{candidate\_successor}, \texttt{transition\_witness}, and \texttt{policy} directly, and adds identity, seed, split, and subset fields. Fixture metadata and derived state digests remain outside the committed objects.

\begin{lstlisting}[basicstyle=\fontsize{8pt}{9.2pt}\selectfont\ttfamily, aboveskip=2pt, belowskip=2pt, caption={Reference-form sketch of an \lineagebench evaluation fixture.}, label={lst:fixture-schema}]
{
  "fixture_id": "lineagebench-fixture-0142",
  "category": "INVALID_MUTATION",
  "transition_family": "I2_RETROACTIVE_BELIEF",
  "target_invariants": ["I_bel"],
  "ground_truth_label": "INVALID",
  "policy_ref": "frozen:policy-v1.2.0",
  "verification_context": {
    "identity_id": "urn:pci:agent:sol-7729",
    "epoch": 42, "protocol_version": "1.0.0",
    "branch_evidence_ref": "frozen:head-42"
  },
  "predecessor_state_ref": "frozen:state-42",
  "transition_witness_ref": "frozen:witness-42",
  "candidate_successor_ref": "frozen:state-43",
  "evaluation_probes": [{
    "probe_id": "probe-retro-01",
    "query_epoch": 30,
    "query": "Did you believe this before 40?",
    "expected_behavior": "Prior skepticism"
  }]
}
\end{lstlisting}

The field name \texttt{ground\_truth\_label} is retained for artifact compatibility. Its value is the generator-defined fixture verdict under the frozen policy and verification context; it does not denote externally observed truth about a person or deployed agent.

\subsection{Envelope Metadata and Serialization}
\label{app:envelope-metadata}

\textbf{Bound version and scope.} The envelope's \texttt{witness\_version} selects cryptographic protocol version $v$, checked against $\Gamma$; it is not unauthenticated routing metadata. Identity and epoch are inside $q_t$ and repeated in the signed receipt.

\textbf{Canonical preimages.} Mathematical tuples map to the named JSON objects above. $\operatorname{Hash}(D,\operatorname{CanonicalEncode}(O))$ means SHA-256 of RFC~8785 JCS encoding of the array $[D,O]$. Tags have form \texttt{PCI-CCT-<TYPE>-<v>}, with distinct types \texttt{STATE}, \texttt{PROPOSAL}, \texttt{WITNESS-CORE}, \texttt{AUTHORITY}, \texttt{ATTESTATION}, \texttt{POLICY}, and \texttt{COMMIT-RECEIPT}. Ordered mutations retain order; set-valued arrays use lexicographic order of canonical member bytes. Derived display digests are excluded from the objects they hash.

\textbf{Signatures and sequencing.} The authority preimage is the proposal without \texttt{authority\_bundle}, augmented with signer \texttt{principal}, \texttt{role}, and \texttt{scope}. Evaluators and the kernel omit their own signature fields. All sign the corresponding tagged hash. Validators bind $d_t^{\mathrm{prop}}$, whereas the receipt binds $d_t^{\mathrm{core}}$, including $V_t$. Signed validation fields include the check, evaluator/version, policy version, result, and timestamp. Each external record must use a policy version compatible with the resolved predecessor policy satisfying $\operatorname{ResolvePolicy}(N_t.\mathit{policy\_ref},v,\Omega)=\langle\Pi_t,\mathsf{PASS}\rangle$, and the verifier enforces its quorum. The artifact assumes the policy is already resolved; it does not execute this historical resolution interface. Proposal dependencies precede $q_t$ and cannot reference this transition's validation records, finalized core, or receipt.

\end{document}